\documentclass[11pt]{article}

\usepackage{fullpage}
\usepackage{times}
\usepackage{amsmath,amsthm}

\usepackage{natbib}

\usepackage{array}
\usepackage{latexsym}
\usepackage{amsmath}
\usepackage{amssymb}
\usepackage{ifthen}

\newcommand{\FULLPAGE}{
	\typeout{Style Option FULLPAGE Version 2 as of 15 Dec 1988}
	\topmargin 0pt
	\advance \topmargin by -\headheight
	\advance \topmargin by -\headsep
	\textheight 8.9in
	\oddsidemargin 0pt
	\evensidemargin \oddsidemargin
	\marginparwidth 0.5in
	\textwidth 6.5in
}

\theoremstyle{plain}
\newtheorem{theorem}{Theorem}[section]
\newtheorem{proposition}[theorem]{Proposition}
\newtheorem{lemma}[theorem]{Lemma}

\newtheorem{corollary}[theorem]{Corollary}

\theoremstyle{definition}
\newtheorem*{definition*}{Definition}
\newtheorem{definition}[theorem]{Definition}

\theoremstyle{remark}
\newtheorem*{remark*}{Remark}
\newtheorem*{property*}{Property}

\newcommand{\xhdr}[1]{{\vspace{2mm}\noindent\bf #1}}

\newcommand{\R} {\ensuremath{\mathbb{R}}} 
\newcommand{\E} {\ensuremath{\mathbb{E}}} 

\newcommand{\eps}{\ensuremath{\epsilon}}

\newcommand{\cel}[1]{{\lceil {#1} \rceil}}
\newcommand{\flr}[1]{{\lfloor {#1} \rfloor}}

\newcommand{\initOneLiners}{%
 	\setlength{\itemsep}{0pt}
	\setlength{\parsep }{0pt}
  	\setlength{\topsep }{0pt}     	
}
\newenvironment{OneLiners}[1][\ensuremath{\bullet}]
    {\begin{list}
        {#1}
        {\initOneLiners}}
    {\end{list}}

\newcommand{\OMIT}[1]{}

\newcommand{\MyTitlePage}{
\makeatletter
 \let\old@oddfoot\@oddfoot
 \renewcommand{\ps@plain}{%
      \renewcommand{\@oddhead}{}%
      \renewcommand{\@evenhead}{\@oddhead}%
      \renewcommand{\@oddfoot}{\ifnum\thepage=0\else\old@oddfoot \fi}
      \renewcommand{\@evenfoot}{\@oddfoot}}
 \makeatother
 \pagestyle{plain}
 \setcounter{page}{0}
}

\newcounter{MyPropertyCounter}

\newcommand{\BKnote}[1]{} 
\newcommand{\YYnote}[1]{} 

\newcommand{\optrev}{{\rho}}
\newcommand{\feas}{{\mathcal{F}}}

\newcommand{\md}{{\mathtt{MD}}}
\newcommand{\mdm}{{\mathtt{MDM}}}

\newcommand{\given}{\,|\,}
\newcommand{\sgn}{{\operatorname{sgn}}}

\newcommand{\val}{{v}}
\newcommand{\vv}{{\phi}}
\newcommand{\ivv}{{\bar{\vv}}}

\newcommand{\opt}{{\mathsf{opt}}}
\newcommand{\indics}{{\mathbf{1}_{S=\opt}}}
\newcommand{\cutoff}{{Z}}
\newcommand{\indith}{{\mathbf{1}_{\cutoff \leq X^*}}}
\newcommand{\indithc}{{\mathbf{1}_{\cutoff > X^*}}}

\begin{document}
\newcommand{\OurTitle}{On the Ratio of Revenue to Welfare in Single-Parameter Mechanism Design}

\title{\OurTitle}
\author{Robert Kleinberg\thanks{Department of Computer Science, Cornell University, Ithaca, NY, USA. Supported in part by NSF Awards CCF-0643934 and AF-0910940,
a Microsoft Research New Faculty Fellowship, and a Google Research Grant.}
\and Yang Yuan\thanks{Department of Computer Science, Cornell University, Ithaca, NY, USA.}}
\date{}

\maketitle
\begin{abstract}
What fraction of the potential
social surplus in an environment 
can be extracted by a revenue-maximizing
monopolist? We investigate this problem
in Bayesian single-parameter environments
with independent private values. The precise
answer to the question obviously depends on
the particulars of the environment: the
feasibility constraint and the distributions
from which the bidders' private values are 
sampled. Rather than solving the problem in
particular special cases, our work aims to
provide universal lower bounds
on the revenue-to-welfare ratio 
that hold under the most general 
hypotheses that allow for non-trivial such bounds.

Our results can be summarized as 
follows. For general feasibility constraints,
the revenue-to-welfare ratio is at least a constant
times the inverse-square-root of the number of agents,
and this is tight up to constant factors. 
For downward-closed feasibility
constraints, the revenue-to-welfare ratio is bounded
below by a constant.
Both results require the bidders' distributions to
satisfy hypotheses somewhat stronger than regularity;
we show that the latter result cannot avoid this
requirement.
\end{abstract}

\section{Introduction}
\label{sec:intro}

When a firm offers a new service with the potential
to bring utility to a set of users, it is intuitive
that the firm should be able to extract a significant
fraction of that utility as profit. Is this intuition
justified by theory? This fundamental
question about the relation between revenue-maximizing
and welfare-maximizing mechanisms is the focus of our
paper.

The answer to our question depends, among other things,
upon which sets of users may potentially be served.
An exemplary case in which the seller's revenue
is only a small fraction of the social surplus is
a \emph{public project}, in which the only two
alternatives are to serve everyone or to serve no one.
As we shall see in Section~\ref{sec:uniform},
for a public project with $n$ agents having i.i.d.\ values
uniformly sampled from $[0,1]$, the optimal mechanism
provides the seller with revenue $\Theta(\sqrt{n})$,
whereas the expected social surplus generated by serving all
agents is $n/2$.

There is a clear economic intuition as to why the seller's
revenue is so limited in the public project setting: there
is no way to deny service to one agent while serving
another, so an agent's bid is unlikely to influence her
own allocation. Accordingly, it is not possible to charge
agents more than a small fraction of their reported value without
creating an incentive for under-reporting.
Pursuing this intuition further, one would expect
the seller to be able to extract a much larger fraction
of the potential social surplus in \emph{downward-closed}
environments, when the decision to deny service to an
agent may be made on an individual basis.

The foregoing discussion inspires some natural
questions about the relation between revenue-maximizing
and welfare-maximizing mechanisms, that refine the guiding
question presented at the start of the paper.
Can the revenue of the
optimal mechanism ever be less than $c/\sqrt{n}$ times the
expected welfare of the efficient allocation, where $n$
is the number of agents and $c$ is a universal constant?
Under what conditions does this revenue-to-welfare ratio
improve to a constant? Our goal in this paper is to answer
these questions for Bayesian single-parameter environments.

A moment's thought reveals that one must place some
restriction on the distributions from which the agents'
values are sampled, to avoid trivialities. For example,
consider a monopolist selling a single item to an agent
whose value is sampled from
the \emph{equal-revenue distribution}, with
cumulative distribution function
satisfying $F(x) = 1-1/x$ for all $x \geq 1$.
As is well known, the seller cannot extract
more than one unit of revenue,
despite the fact that allocating the item
yields infinite expected welfare
in this case. Thus, even in the extremely
simple setting a single-item auction with one
agent, the seller is not guaranteed any
constant fraction of the social surplus unless
we make further assumptions about the distributions
of agents' values.

A theme running through many of our results is
that the foregoing type of distribution --- one
that prevents the seller in a single-item auction
from extracting a constant fraction of the buyer's
expected value --- is essentially the only type of
distribution that must be excluded in order to
obtain strong lower bounds on the
revenue-to-welfare ratio
under arbitrary feasibility
constraints. To make this more precise, for a non-negative
real-valued random variable $X$ with cumulative distribution
function $F(x)$, let $\optrev(X)$ denote
the seller's optimal revenue when selling an item
to a single agent with private value $X$:
\[
\optrev(X) = \sup_{p \geq 0} \{ p \cdot (1 - F(p) ) \}.
\]
We now define the following two properties of
a distribution.
\begin{definition} \label{def:c-bounded}
For any number $c>0$, we say a random variable $X$ is
\emph{$c$-bounded} if it satisfies $c \cdot \optrev(X) \geq \E[X]$,
and it is \emph{strongly $c$-bounded}
if $\Pr(c \cdot \optrev(X) \geq X) = 1$.
\end{definition}
In other
words, a buyer's value distribution is $c$-bounded if her
expected value is at most $c$ times the revenue that a seller
can earn when selling one item to her, and it is strongly
$c$-bounded if her value is \emph{never} more than $c$ times
the seller's optimal revenue.
Having made these definitions, we can state our main
results. All of them pertain to Bayesian
single-parameter environments in which $n$
agents have independent private values
and the feasibility constraint is specified by
a set system $\feas \subseteq 2^{[n]}$ denoting
the sets of agents that may be simultaneously served.
\begin{theorem} \label{thm:sqrt-approx}
If $\feas$ is arbitrary, and all agents have
strongly $c$-bounded distributions, then the
revenue of the optimal mechanism is at least
$1/(96 c \sqrt{n})$ times the expected welfare
of the efficient allocation. For public project
mechanisms, the same conclusion holds under the
weaker hypothesis that the distributions are
$c$-bounded.
\end{theorem}
The following theorem refers to \emph{hyper-regular
distributions}, a mild specialization of regular
distributions whose definition we defer to Section~\ref{sec:prelims}.
All hyper-regular distributions are regular, and while the
converse is not true, it is the case that most of the commonly
cited examples of regular distribution --- including monotone
hazard rate (MHR)
distributions and Pareto distributions --- are hyper-regular.
See the paragraph following Definition~\ref{def:hyper-reg} 
for further discussion of this point.
\begin{theorem} \label{thm:const-approx}
If $\feas$ is downward-closed, and all agents
have $c$-bounded hyper-regular distributions,
then the revenue of the optimal mechanism is at
least $1/c$ times the expected welfare
of the efficient allocation.
\end{theorem}
We further show that the assumption of hyper-regularity
is unavoidable in Theorem~\ref{thm:const-approx}, even
when dealing with single-item auctions. We give an explicit
example of a regular (but not hyper-regular) distribution
$F$ such that as $n \to \infty$, the ratio of
the optimal revenue to the maximum bid tends to zero in
a single-item auction with $n$ i.i.d.\ bidders sampling
values from $F$.

To derive our results, we use a mix of techniques from
economics and probability theory. Not surprisingly, we
rely heavily on Myerson's Lemma that the expected revenue of a
mechanism equals its expected virtual surplus. We then face the
task of proving lower bounds on the expected virtual surplus of
the optimal mechanism. It turns out that this task is closely
tied to proving \emph{anti-concentration} inequalities for
sums of independent random variables, i.e.\ inequalities asserting
that the sum is unlikely to be too \emph{close} to its
expected value. We derive an anti-concentration inequality
suited to our application by generalizing Erd\H{o}s's proof
of the Littlewood-Offord Theorem~\cite{erdos,littlewood-offord}.
This inequality constitutes the main technical ingredient
underlying Theorem~\ref{thm:sqrt-approx}. To obtain
Theorem~\ref{thm:const-approx} we generalize a different
tool from probability theory, namely Chebyshev's Integral
Inequality.

\xhdr{Related work.}
Many prior papers address relationships between 
revenue-maximizing and welfare-maximizing mechanisms 
in Bayesian settings. All of these papers are 
thematically related to our work, and some of them
contain theorems that directly imply bounds on
the revenue-to-welfare ratio for special cases
of the settings considered here, though usually
as a side effect of attacking other questions.
For example, the famous work of 
\citet{Bulow-Klemperer} shows that the
revenue of the Vickrey single-item
auction with $n+1$ i.i.d.\ bidders exceeds
that of the optimal single-item auction with
$n$ i.i.d.\ bidders drawn from the same distribution,
provided the distribution is regular. (Note the 
constrast with our work: theirs relates the \emph{revenue}
of a VCG auction to that of an
optimal auction, whereas our work relates the 
\emph{efficiency} of a VCG auction to the
\emph{revenue} of an optimal auction.)
Drawing inspiration from Bulow and Klemperer
while significantly expanding upon their
techniques, \citet{Dhang2010} designed
\emph{single sample mechanisms} and proved ---
under various hypotheses on the feasibility
constraints and the distributions --- that their
mechanism's revenue approximates that of the optimal
mechanism. All of the environments considered in 
their paper have downward-closed feasibility 
constraints, unlike our paper that also addresses
general feasibility constraints. 
Of particular relevance to our work
is Theorem 3.10 of \citep{Dhang2010}, 
which directly bounds the revenue-to-welfare
ratio of the ``VCG with lazy reserves'' (VCG-L)
mechanism in downward-closed environments with
MHR distributions. Our 
Theorem~\ref{thm:const-approx} can be 
seen as a generalization of their Theorem 3.10
from MHR distributions to hyper-regular
distributions.\BKnote{We should actually put a 
corollary after the proof of 
Theorem~\ref{thm:const-approx},
in Section 5, noting that the 
same approximation bound applies 
when one uses VCG-L instead of an
optimal mechanism.}

Other extensions of the Bulow-Klemperer Theorem
in recent years have contributed to the 
literature on relations between revenue-maximizing
and welfare-maximizing auctions.
For example, \citet{HR09} consider duplicating
\emph{each} bidder, and they bound the ratio
between  the
revenue of the VCG mechanism  in the 
``duplicated environment'' and that
of the optimal mechanism in the 
original environment; this technique
is then used to imply that simple
mechanisms that modify VCG by adding
reserve prices can approximate the
revenue of the optimal mechanism.
Extending Bulow-Klemperer in a different direction, \citet{AGM09}
show that adding $O(\log n)$ additional bidders
to Myerson's mechanism (in an i.i.d.\ m.h.r.\ single-item
environment) is necessary and sufficient to achieve an
expected welfare guarantee that matches that of the 
VCG mechanism with the original $n$ bidders.

Other papers contributing to the literature on 
relationships between welfare-maximizing and
revenue-maximizing mechanisms are
\citep{Daskalakis2011}, which presents
auctions that simultaneously achieve good
revenue and efficiency for 
single-item environments, 
and \citep{hajek2010},
which considers the efficiency loss
in revenue-maximizing mechanisms.

Our paper is not the first to use the 
Littlewood-Offord Theorem and its
generalizations to bound the revenue
of mechanisms. A different generalization
of Littlewood-Offord was applied by 
\citet{karlin2012} to the analysis of 
prior-free mechanisms.

\section{Preliminaries}
\label{sec:prelims}

\xhdr{Single-parameter Bayesian mechanism design.}
In a standard single-parameter Bayesian mechanism design setting,
there are $n$ bidders or agents, each with a
private value $\val_i, i=1,\ldots,n$,
denoting the value of agent $i$ for receiving service.
We will denote the cumulative distribution function
of $\val_i$ by $F_i$, and when we assume that $\val_i$
has a density function we will denote the density function
by $f_i$.

A general feasibility environment is specified by
a set $\feas \subseteq 2^{[n]}$ denoting the
feasible sets of bidders that can be simultaneously served.
We call $\feas$ the \emph{feasibility constraint} of the
environment. We say $\feas$ is \emph{downward-closed}
if every subset of a feasible set is feasible.

A mechanism is a pair $(A,p)$ consisting of an
\emph{allocation function} $A: \R^n \to \{0,1\}^n$
and a \emph{payment function} $p : \R^n \to \R^n$.
Both functions may possibly be randomized.
The input to both functions is a vector of bids.
The function $A$ determines the set of agents
who will be served;
thus we require that $\{i : A_i(b)=1\}$ belongs to
$\feas$ for every possible bid vector $b$. 
The payment function $p$ determines how much
each agent will pay.
Agents are risk-neutral
and have quasi-linear utility: an agent with
value $\val_i$ who is served with probability
$\pi_i$ and pays $p_i$ has utility $\pi_i \val_i - p_i$.

The \emph{expected revenue} (or simply \emph{revenue})
of a mechanism is $\E [ \sum_{i=1}^n p_i(b) ]$
where $b$ is the random bid vector in some equilibrium
of the mechanism. Its \emph{expected welfare}
(or simply \emph{welfare}) is
$\E[ \sum_i A_i(b) \val_i ]$,
the expected sum of values of the agents served.
In both cases, the expectation is over the
randomness in the agents' private values,
as well as the randomness (if any) in their
choice of bids and in the mechanism's choice
of allocations and payments.
All mechanisms in this paper are assumed to be
\emph{ex post individually rational}, meaning
that agents are never charged an amount
exceeding their bid.



\xhdr{Probability distributions.}
When $X$ is a random variable,
we denote by $X^+=\max\{0,X\}$ the ``positive part'' of $X$,
and by $X^-=\min\{0,X\}$ the ``negative part'' of $X$.

The hazard rate of a distribution is defined as $h(x)=\frac{f(x)}{1-F(x)}$, and
a monotone hazard rate (MHR) distribution is one whose  hazard rate
is non-decreasing.
The virtual valuation function corresponding to distribution $F$ is
$\vv(x)=x-\frac{1}{h(x)}$.
Distributions with non-decreasing
virtual valuation function are called regular distributions.
In the sequel, we will use the following strengthening of the
regularity property.
\begin{definition}[Hyper-regular Distribution] \label{def:hyper-reg}
A hyper-regular distribution is a regular distribution with non-decreasing $\frac{\phi(x)}{x}$.
\end{definition}
Most of the common examples of regular distributions are actually
hyper-regular. For example,
it is easy to see that all MHR distributions are hyper-regular.
Also, Pareto distributions having cumulative density function
$F(x)=1-x^{-\alpha}$, where $\alpha>1$ (a necessary condition
for the distribution to be regular, and also for it 
to have finite expected value) are hyper-regular.
Not all regular distributions are hyper-regular;
for example, the distribution specified by
$F(x-\delta)=1- \frac{1}{x\ln^2 x}$, where
$\delta \ln^2 \delta=1$, is not hyper-regular.
We will return to this distribution at the end
of Section~\ref{sec:dc}.

\xhdr{Myerson's lemma.}

\citet{myerson1981} gave a connection between the expected revenue and the
expected virtual surplus.
\begin{lemma}[Myerson's Lemma]
In a truthful mechanism $(A,p)$ 
the expected payment $p_i$ of agent $i$ with virtual valuation
function $\phi_i$ satisfies:
\[
\E[p_i(\mathbf{v})]= \E [\phi_i(v_i) \cdot A_i(\mathbf{v})]\]
The equality holds even when the bids of other bidders $v_{-i}$ are fixed.
\end{lemma}
Thus, when virtual surplus maximization induces a monotone allocation rule, this allocation rule maximizes revenue. This criterion is always satisfied
when bidders' values are drawn from regular distributions. When
the distributions are not regular, Myerson provides a workaround:
an \emph{ironed virtual valuation function} $\ivv_i$ for each
bidder, which is always monotone, such that Myerson's Lemma
continues to hold provided that the allocation rule is constant
on any interval in which the bidder's ironed virtual value is constant.
Ironed virtual surplus maximization induces a monotone
allocation rule, and a mechanism with this allocation rule
maximizes revenue.

To maximize the welfare, we can use the well-known VCG mechanism.
In this paper we also use a variation of the VCG mechanism called
``VCG with lazy reserves'', or simply VCG-L~\citep{Dhang2010}, which operates as follows:
\begin{enumerate}\itemindent=15pt
\item Run the VCG mechanism to obtain a preliminary winning set $P$.
\item Remove all the bidders $i\in P$ with $v_i< r_i$, where $r_i=\phi_i^{-1} (0)$ is the reserve price for the bidder $i$.
\item Charge each winning bidder $i$ the larger of $r_i$ and its VCG payment in the first
step.
\end{enumerate}

\section{Warm-up: Identical uniform distributions}
\label{sec:uniform}

As a prelude to our main results, we devote this section
to bounding the revenue-to-welfare ratio when the bids
are i.i.d.\ uniform samples from [0,1]. The results
in this section will be completely subsumed by subsequent
theorems, but they have much simpler proofs that serve
to illustrate the main ideas underlying our later results
while highlighting the technical challenges that must be
overcome in order to prove those more general results.

The uniform distribution on [0,1] has a very simple 
virtual valuation function. We have $F(x)=x$ and $f(x)=1$
for all $x \in [0,1]$, and so
\[
\vv(x) = x - \frac{1-F(x)}{f(x)} = x - (1-x) = 2x-1.
\]
The following simple consequence is important for our analysis.
\begin{equation} \label{eq:unif-vv}
\mbox{If $x$ is uniformly distributed in $[0,1]$ then
$\vv(x)$ is uniformly distributed in $[-1,1]$.}
\tag{*}
\end{equation}
Let us first use these observations to 
derive an asymptotic expression for
the revenue-to-welfare ratio for a 
public project with $n$ i.i.d.\ uniform [0,1] bids.
The allocation that provides service to all bidders
also maximizes welfare, so the expected welfare of
the efficient allocation is simply:
$\E[x_1+\cdots+x_n] = \tfrac{n}{2}.$
The virtual surplus is maximized by serving 
everyone if $\phi_1(x_1)+\cdots+\phi_n(x_n) \geq 0$,
and otherwise by serving no one. Therefore,
the optimal mechanism's revenue is
$\E[(\phi_1(x_1) + \cdots + \phi_n(x_n))^+\}]$.
An asymptotic formula for this expression 
can readily be computed
using the Central Limit Theorem. The random
variables $\phi_i(x_i)$ are i.i.d.\ uniform
samples from $[-1,1]$, so they have
mean zero and variance $\sigma^2 = \frac13$.
Consequently the random variable
$n^{-1/2} \sum_{i=1}^n \phi_i(x_i)$
converges in distribution to $\mathcal{N}(0,1/3)$,
and thus
\[
\lim_{n \to \infty} \left\{
\tfrac{1}{\sqrt{n}} 
\E \left[ \left( \phi_1(x_1) + \cdots + \phi_n(x_n) \right)^+ \right]
\right\} = \tfrac{1}{\sqrt{6 \pi}} \int_0^\infty t e^{-t^2/2} \, dt = 
\tfrac{1}{\sqrt{6 \pi}}.
\]
Recalling that the expectation of the maximum welfare in 
this case is $n/2$, we see that the 
revenue-to-welfare ratio is asymptotic to
$\sqrt{\frac{2}{3 \pi n}}$,
and in particular it is $\Theta(n^{-1/2})$.
\OMIT{
 In principle
one could apply the Berry-Esseen Theorem, which strengthens
the Central Limit Theorem by bounded the error in the
asymptotic expression for finite $n$; this would allow us
to derive an
explicit constant $\alpha$ such that the revenue-to-welfare 
ratio is bounded below by $\alpha n^{-1/2}$ for all $n$, and
not just for large $n$. The analysis of public projects
in Section~\ref{sec:pub_proj} will present an alternative
means of deriving such a constant $\alpha$.
}

Let us now generalize to arbitrary feasibility constraints.
Intuitively, it seems that the
revenue-to-welfare ratio should be minimized 
by the public project environment, for the 
reasons articulated in the introduction.
The following proposition confirms that this
intuition is valid, at least up to a constant 
factor.\footnote{In fact, the argument given in the proof shows
that the public project minimizes the revenue-to-welfare
ratio up to a factor of 2. It is an interesting open
question whether the revenue-to-welfare ratio is 
\emph{precisely} minimizes by the public project.}
\begin{proposition} \label{prop:unif-gen}
For a Bayesian single-parameter environment with
$n$ i.i.d.\ bidders having uniform $[0,1]$ values,
and a general feasibility constraint $\feas$,
the revenue-to-welfare ratio is always
at least $\Omega(n^{-1/2})$.
\end{proposition}
\begin{proof}
Among all feasible sets, let 
$S^*$ be one with maximum cardinality,
$k = |S^*|$. We will show that the revenue-to-welfare
ratio is $\Omega(k^{-1/2})$, from which the 
proposition follows \emph{a fortiori}.

As the bidders' values are never
greater than 1, the welfare of the efficient allocation 
is never greater than $k$. 
Consider a mechanism $\mathcal{M}$ 
which maximizes revenue subject to the 
constraint that the set of agents served is always
either $\emptyset$ or $S^*$. This is simply
an optimal mechanism for a public project with agent set $S^*$, 
so we have already calculated that its revenue is
$\Theta(k^{1/2})$. The revenue of the optimal
mechanism is at least as great as that of $\mathcal{M}$,
hence the revenue-to-welfare ratio is 
$\Omega(k^{1/2} / k) = \Omega(k^{-1/2})$,
as claimed.
\end{proof}

When the feasibility constraint is downward closed,
and bids are i.i.d.\ uniform in [0,1],
an even easier argument establishes that the 
revenue-to-welfare ratio is $\Omega(1)$.
\begin{proposition} \label{prop:unif-dc}
For a Bayesian single-parameter environment with
$n$ i.i.d.\ bidders having uniform $[0,1]$ values,
and a downward-closed feasibility constraint $\feas$,
the revenue-to-welfare ratio is always
at least $\frac14$.
\end{proposition}
\begin{proof}
As before, define $S^* \in \feas$ to be a feasible
set of maximum cardinality, and let $k = |S^*|$.
The welfare of the efficient allocation is never
greater than $k$, and we will prove that the 
revenue of the optimal mechanism is at least 
$k/4$.

Let $\mathcal{M}'$ be the mechanism
that maximizes revenue subject to the 
constraint that the set of agents served
is always a subset of $S^*$. By Myerson's
Lemma, the expected revenue of $\mathcal{M}'$
is simply $\sum_{i \in S^*} \E[\phi_i(x_i)^+]$.
Recalling that $\phi_i(x_i)$ is uniformly
distributed in $[-1,1]$, we see that 
$\E[\phi_i(x_i)^+] = \frac14$ for each $i$,
and the result follows.
\end{proof}

As we aim to extend these results
to general distributions, it is worthwhile
to reflect on the aspects of the proofs that
were specific to the uniform distribution.
\begin{enumerate}
\item Our analysis of the revenue-to-welfare
ratio of the public project hinged on deriving
the asymptotic lower bound 
$\E[(\phi_1(x_1)+\cdots+\phi_n(x_n))^+] = 
\Theta(\sqrt{n})$. We achieved this using
the Central Limit Theorem. To extend this
step to more general --- and not necessarily
identical --- distributions, we require what
might be called \emph{anti-concentration} inequalities
for sums of independent random variables.
Versions of the Central Limit Theorem for
non-identical distributions exist, but they
are not general enough for our purposes.
(For instance, they require upper bounds
on the second moments, whereas we do not.)
Instead we will generalize a different anti-concentration
inequality, the Littlewood-Offord Theorem.
\item In the proofs of both propositions, we bounded
the expected welfare of the efficient allocation by
the cardinality of the maximum feasible set. This 
very simple upper-bounding technique was effective
because the uniform distribution is strongly
$c$-bounded for $c=2$. (The expected welfare
of any set of agents is at least half of its cardinality.)
When dealing with distributions that are not
strongly $c$-bounded, we need to develop a different
technique for upper-bounding the expected welfare
of the efficient allocation.\BKnote{A simple description of
the ``completely new technique'' would be convenient here.}
\end{enumerate}

\section{General feasibility constraints}
\label{sec:gen}

In this section, we consider arbitrary
feasibility constraints with $n$ agents
and extend
the $\Omega(n^{-1/2})$ lower bound
on the revenue-to-welfare ratio
(Proposition~\ref{prop:unif-gen})
from i.i.d.\ uniform
bids to more general distributions.
As noted at the end of Section~\ref{sec:uniform},
the key to proving such an extension is to
derive an inequality asserting that the distribution
of a sum
of independent random variables cannot be
too tightly concentrated around its expected
value. We first derive a
suitably general inequality in
Subsection~\ref{sec:lob}, and
we apply this inequality in
the following subsections.

\subsection{A generalization of the Littlewood-Offord Theorem}
\label{sec:lob}

A beautiful ``anti-concentration'' inequality
for independent random variables was proven
by~\citet{littlewood-offord} and strengthened
by~\citet{erdos}.
\begin{theorem}[\citep{littlewood-offord,erdos}] \label{thm:littlewood-offord}
For any real numbers $x_1,\ldots,x_n \geq 1$ and any
half-open interval $I$ of length 2,
the number of sums
$\sum_{i=1}^n \eps_i x_i$ that belong to $I$
as the vector
$(\eps_1,\ldots,\eps_n)$ ranges over $\{\pm 1\}^n$,
is at most $\binom{n}{\lfloor n/2 \rfloor}$.
\end{theorem}

In this section we
present a more general anti-concentration
inequality for sums of independent random variables.
To state our generalization, we must first specify a
few notations concerning deviations of random variables.

\begin{definition} \label{def:dev}
For a random variable $X$, we define its \emph{median} $m(X)$ to be
any number such that $\Pr(X < m(X)) \leq 1/2$
and $\Pr(X > m(X)) \leq 1/2$. We will denote
the absolute deviation of $X$
from its mean and median by
 \begin{align*}
 \md(X) &= \E |X - \E X | \\
 \mdm(X) &= \E |X - m(X)|.
 \end{align*}
Note that if there is more than one number $m(X)$
satisfying the definition of the median of $X$,
then the value of $\mdm(X)$ is independent of the
choice of $m(X)$.
\end{definition}

The following simple relations between $\md(X)$ and
$\mdm(X)$ are proven in Appendix~\ref{app:prob} 

\begin{lemma} \label{lem:mdm}
For any random variable $X$ and any
constant $a$,
\begin{equation} 
\mdm(X-a) = \mdm(a-X) = \mdm(X) 
\quad \mbox{and} \quad 
\md(X-a) = \md(a-X) = \md(X).
\end{equation}
Furthermore,
\begin{equation} 
\mdm(X) \leq \md(X) \leq 2 \mdm(X).
\end{equation}
\end{lemma}

Our first anti-concentration result is stated in the following
proposition, whose proof is also deferred to Appendix~\ref{app:prob}.
\begin{proposition} \label{prop:lob.2}
If $X_1,\ldots,X_n$ are independent random variables
and $\mdm(X_i) \geq 1$ for all $i$, then
$\mdm(X_1+\ldots+X_n) \geq \frac{1}{12} \sqrt{n}.$
\end{proposition}

We leverage the proposition to derive the
following result.

\begin{theorem} \label{thm:lob.3}
Let $Y_1,\ldots,Y_n$ be any $n$-tuple of
independent random variables,
each with expectation zero.\BKnote{TO DO:
modify the theorem statement and proof to provide a
similar conclusion when the $Y_i$ are allowed
to have strictly positive expectations.}
Let $z_i = \E[Y_i^+]$ for $i=1,\ldots,n$.
Then
\begin{equation} \label{eq:thm.lob.3}
\E[(Y_1+\cdots+Y_n)^+] \geq \frac{z_1+\cdots+z_n}{48 \sqrt{n}}.
\end{equation}
\end{theorem}
\begin{proof}
Assume, without loss of generality, that
$z_1 \geq z_2 \geq \cdots \geq z_n$. Also
assume that $\max_{1 \leq k \leq n} \{z_k \sqrt{k}\} = 1$.
The latter assumption is without loss of generality
because we can rescale all the random variables
$Y_1,\ldots,Y_n$ by the same positive scalar without
affecting the lemma's hypotheses or conclusion.

Our assumption that $\max \{z_k  \sqrt{k}\} = 1$
implies that $z_k \leq k^{-1/2}$ for all $k$, hence
$$
z_1+ \cdots+z_n \leq \sum_{k=1}^{n} k^{-1/2} < 2 \sqrt{n}.
$$
Consequently, the value of $(z_1+...+z_n)/(12 \sqrt{n}))$
is bounded above by $\frac16$.
If we can show that the expected value of $|Y_1+\cdots+Y_n|$
is bounded below by a constant, we are done,
since the relation
$\E[(Y_1+\cdots+Y_n)^+] = \frac12 \E |Y_1+\cdots+Y_n|$
holds for the mean-zero random variable $Y_1+\cdots+Y_n$.

We know, from Lemma~\ref{lem:mdm},
that  for all $i$,
$\mdm(Y_i) \geq \frac12 \E |Y_i| = z_i$.
Applying Proposition~\ref{prop:lob.2} to the random sum
$Y_1+\cdots+Y_k$,
it follows that the expected absolute value of that sum is
at least $\frac{1}{12} z_k \sqrt{k} = \frac{1}{12}$.

Next we show that
            $\E |Y_1+\cdots+Y_n| \geq \E |Y_1+...+Y_k|.$
We have
\begin{align}
\label{eq:thm.lob.2}  \E [ \sgn(Y_1+\cdots+Y_n) \cdot (Y_1+\cdots+Y_n) ] &= \E |Y_1+\cdots+Y_n| \\
 \E [ \sgn(Y_1+\cdots+Y_k) \cdot (Y_1+\cdots+Y_n) ] &= \E |Y_1+\cdots+Y_k| 
\nonumber \\
\label{eq:thm.log.3} 
& \quad \;\; + \E [ \sgn(Y_1+\cdots+Y_k) \cdot (Y_{k+1}+\cdots+Y_n) ] \\
& = \E |Y_1+ \cdots +Y_k|
\end{align}
where the last equality holds because $Y_1+\cdots+Y_k$ is independent of
$Y_{k+1}+\cdots+Y_n$, and the latter has zero expected value.

The left side of~\eqref{eq:thm.lob.2} is greater than or equal to the left side of~\eqref{eq:thm.lob.3}, because the inequality
$$  [\sgn(Y_1+\cdots+Y_n) - \sgn(Y_1+\cdots+Y_k)] \cdot (Y_1+\cdots+Y_n) \geq 0 $$
holds for all values of $Y_1,\ldots,Y_n$.
Indeed, whenever the quantity $[\sgn(Y_1+\cdots+Y_n)-\sgn(Y_1+\cdots+Y_k)]$
is nonzero, it has the same sign as $Y_1+\cdots+Y_n$.
Combining previous steps, we obtain
\[
\E |Y_1+\cdots+Y_n| \geq \E |Y_1+\cdots+Y_k| \geq \frac{1}{12}
\geq \frac{z_1+\cdots+z_n}{24 \sqrt{n}},
\]
and the theorem follows since
$\E[(Y_1+\cdots+Y_n)^+] = \frac12 \E |Y_1+\cdots+Y_n|.$
\end{proof}

\begin{theorem} \label{thm:lob.4}
Let $Y_1,\ldots,Y_n$ be any $n$-tuple of
independent random variables with
positive expectations $y_1,\ldots, y_n$.
Let $z_i = \E[Y_i^+]$ for $i=1,\ldots,n$.
Then
\begin{equation} \label{eq:thm.lob.4}
\E[(Y_1+\cdots+Y_n)^+] \geq \frac{z_1+\cdots+z_n}{96 \sqrt{n}}.
\end{equation}
\end{theorem}
\begin{proof}
We write $Y_i = Y'_i + y_i$ and $z_i = z'_i + y_i$,
for each $i$. Then the expectation of $Y'_i$ is zero, and
according to Theorem \ref{thm:lob.3}, we know that
\[E[(Y_1+\cdots+Y_n-\sum_i y_i)^+]\geq \frac{z_1+\cdots+z_n-\sum_i y_i}{48\sqrt{n}}\]
Note $\sum_i y_i$ is positive, so the inequality above gives
a lower bound on $E[(Y_1+\cdots+Y_n)^+]$.

In addition, we know
\[E[(Y_1+\cdots+Y_n)^+]\geq
\sum_i y_i,\]
as otherwise the expectation of $Y_1+\cdots+Y_n$ would be less than $\sum_i y_i$, a contradiction.

Thus,
\begin{align*}
E[(Y_1+\cdots+Y_n)^+] &\geq \max
\left\{\frac{z_1+\cdots+z_n-\sum_i y_i}{48\sqrt{n}},
\sum_i y_i \right\}\\
&\geq \frac12 \left( \frac{z_1+\cdots+z_n-\sum_i y_i}{48\sqrt{n}}+\sum_i y_i
\right)
\geq \frac{z_1+\cdots+z_n}{96\sqrt{n}}
\end{align*}
\end{proof}


\subsection{Public projects}
\label{sec:pub-proj}

In this section we analyze the revenue-to-welfare
ratio for a public project with $c$-bounded
distributions, as a step toward analyzing
environments with general feasibility constraints.

\begin{proposition} \label{prop:pub-proj}
In a public project environment whose $n$ agents
have independent $c$-bounded distributions, the
revenue of the optimal mechanism is at least
$1/(96 c \sqrt{n})$ times the expected welfare
of the efficient allocation.
\end{proposition}
\begin{proof}
For each agent $i$, recall that $\ivv_i$ denotes
the agent's ironed virtual valuation function and that
$\rho(\val_i) = \E[\ivv_i(\val_i)^+]$
denotes the maximum revenue that a
monopolist can obtain by a selling a single
item to agent $i$; denote this number by $\rho_i$
henceforth in the proof. Our assumption
that the value distribution $F_i$ is
$c$-bounded implies that $\E \val_i \leq c \rho_i$,
so the expected welfare of the efficient allocation
is bounded by $c (\rho_1 + \cdots + \rho_n)$.

Applying Theorem~\ref{thm:lob.4} to the
random variables $Y_i = \ivv_i(\val_i)$
\BKnote{Not quite! The expectation of $Y_i$ might
be strictly positive, but this should only help us.
The statement and proof of Theorem~\ref{thm:lob.3}
need to be modified so that this step becomes easy.}
we conclude that
\begin{equation} \label{eq:pub-proj}
\E[(Y_1+\cdots+Y_n)^+] \geq \frac{\rho_1+\cdots+\rho_n}{96 \sqrt{n}}.
\end{equation}
This completes the proof, since the left side
is the optimal mechanism's expected revenue.
\end{proof}

\subsection{Strongly $c$-bounded distributions}
\label{sec:gen-scb}

In this section we prove the first part of
Theorem~\ref{thm:sqrt-approx}, which deals
with arbitrary feasibility constraints
and strongly $c$-bounded distributions.
The second part of the theorem, which deals
with public projects, was already proven
in Proposition~\ref{prop:pub-proj} above.

\begin{proposition} \label{prop:sqrt-approx}
If $\feas$ is arbitrary, and all agents have
strongly $c$-bounded distributions, then the
revenue of the optimal mechanism is at least
$1/(96 c \sqrt{n})$ times the expected welfare
of the efficient allocation.
\end{proposition}
\begin{proof}
As in the preceding proof, let
$Y_i = \ivv_i(\val_i)$ and
$\rho_i = \E [ Y_i^+ ]$
for each agent $i$.
Our assumption that
the value distribution $F_i$ is strongly
$c$-bounded implies that $v_i$ is never
greater than $c \rho_i$.
For any set of agents $S$, let
$\rho(S) = \sum_{i \in S} \rho_i$
and define $S^*$ to be a feasible set that
maximizes $\rho$. Let $k = |S^*|$.
We will show that the revenue-to-welfare
ratio is at least $1/(96 c \sqrt{k})$, from which the
proposition follows \emph{a fortiori}.

As bidder $i$'s value never exceeds $c \rho_i$,
the value of any allocation $S \in \feas$ never
exceeds $c \rho(S)$, which is in turn bounded
above by $c \rho(S^*)$. Hence $c \rho(S^*)$ is
an upper bound on the expected welfare of the
efficient allocation.

Consider a mechanism $\mathcal{M}$
which maximizes revenue subject to the
constraint that the set of agents served is always
either $\emptyset$ or $S^*$. This is simply
an optimal mechanism for a public project with
agent set $S^*$, so its expected revenue
is $\E \left[ \left( \sum_{i \in S^*} Y_i \right)^+ \right]$.
Theorem~\ref{thm:lob.4} guarantees that
\[
\E \left[ \left( \sum_{i \in S^*} Y_i \right)^+ \right] \geq
\frac{\rho(S^*)}{96 c \sqrt{k}},
\]
and the proof is complete.
\end{proof}

\section{Downward-closed environments}
\label{sec:dc}

In this section we consider the revenue-to-welfare ratio
of environments with downward-closed feasibility constraints.
Our main result shows that the optimal mechanism's revenue
is a $\Omega(1/c)$ fraction of the expected welfare of the
efficient allocation, when the distributions are $c$-bounded
and hyper-regular. Recall that a hyper-regular distribution
is one such that $\vv(x)/x$ is a non-decreasing function of $x$,
where $\vv$ denotes the virtual value function.

\subsection{An inequality for monotonic functions of a random variable}
\label{sec:cii}

As in Section~\ref{sec:gen}, a probabilistic inequality lies
at the heart of our main result. In this case, the inequality
in question is a generalization of Chebyshev's Integral
Inequality~\citep{chebyshev}, which asserts that for two
monotonically non-decreasing function $f,g$ on an interval $(a,b)$,
\[
\frac{1}{b-a} \int_a^b f(x) g(x) \, dx \geq
\left[ \frac{1}{b-a} \int_a^b f(x) \, dx \right]
\left[ \frac{1}{b-a} \int_a^b g(x) \, dx \right].
\]
Our generalization is the following lemma.
\begin{lemma} \label{lem:chebyshev}
Suppose $f, g, h$ are three functions of a
real number $x$, such that $f,g$ are both
monotonically non-decreasing, and $h(x) \geq 0$
for all $x$. Then for any random variable $X$
such that $\E[h(X)] > 0$, we have:
\[
\frac{\E [ f(X) g(X) h(X) ]}{\E [ g(X) h(X) ]} \geq
\frac{\E [ f(X) h(X) ]}{\E [ h(X) ]}.
\]
\end{lemma}
The version of Chebyshev's Integral Inequality
stated above is obtained by setting $h(x)=1$
and taking $X$ to be uniformly distributed in $(a,b)$.
We now present the proof of the lemma.
\begin{proof}
Let $c = \frac{\E[ f(X) h(X) ]}{[\E[ h(X) ]}.$
The inequality
$\frac{\E [ f(X) g(X) h(X) ]}{\E [ g(X) h(X) ]} \geq c$
is equivalent to $$\E[ (f(X) - c) g(X) h(X) ] \geq 0,$$
which we now prove.
Since $f(x)$ is non-decreasing,
there is a value $x_0$ such that
$f(x) \leq c$ for $x < x_0$ and
$f(x) \geq c$ for $x > x_0$.
Since $g(x)$ is also non-decreasing, the inequality
$(f(x) - c) \cdot (g(x) - g(x_0)) \geq 0$
holds for all $x$. Rewrite this inequality
as $(f(x)-c) \cdot g(x) \geq (f(x)-c) \cdot g(x_0)$
and use it to deduce:
\[
\E[ (f(X) - c) g(X) h(X) ] \geq
\E[ (f(X) - c) g(x_0) h(X) ] \geq
g(x_0) \cdot \left\{ \E[f(X) h(X)] - c \E[ h(X) ] \right\} = 0,
\]
which proves the lemma.
\end{proof}

\subsection{A revenue-to-welfare bound for hyper-regular
distributions}
\label{sec:hyper-reg}

\begin{theorem} \label{thm:hyper-reg.1}
If $\feas$ is downward-closed, and all agents
have $c$-bounded hyper-regular distributions,
then the revenue of the optimal mechanism is at
least $1/c$ times the expected welfare
of the efficient allocation.
\end{theorem}
\begin{proof}
Fix any allocation $S$, and let $\indics$ denote
the indicator random variable of the event that
$S$ is the welfare-maximizing allocation.
For any bidder $i$ let
$$
g_i(x) =
\E[ \indics \given \val_i=x ] =
\Pr(S=\opt \given \val_i=x).
$$
Note that for every $i \in S$,
the function $g_i(x)$ is non-decreasing
for the simple reason
that if $i$ is a bidder in the welfare-maximizing
set and her value increases, the welfare-maximizing
set remains the same.

For any $i \in S$ let us apply
Lemma~\ref{lem:chebyshev} to
the functions $f(x) = \vv_i(x)^+/x, \,
g(x) = g_i(x), \, h(x) = x$, and the
random variable $X = \val_i$. The function
$f(x)$ is non-decreasing because bidder $i$
has a hyper-regular distribution, and
the function $g(x)$ was proven to be
non-decreasing in the first paragraph
of this proof. Thus, we conclude that
\begin{equation} \label{eq:hyper-reg.1}
\frac{\E[\vv_i(\val_i)^+ g_i(\val_i)]}{\E[\val_i g_i(\val_i)]}
\geq
\frac{\E[\vv_i(\val_i)^+]}{\E[\val_i]}.
\end{equation}
The right side of the inequality is at least $\frac{1}{c}$, because
$\val_i$ is sampled from a $c$-bounded distribution.
To interpret the left side, recall the definition of $g_i$.
We have:
\begin{equation} \label{eq:hyper-reg.2}
\frac{\E[\vv_i(\val_i)^+ g_i(\val_i)]}{\E[\val_i g_i(\val_i)]}
=
\frac{\E[\vv_i(\val_i)^+ \E[ \indics \given \val_i]]}
{\E[\val_i \E[ \indics \given \val_i]]}
=
\frac{\E[\vv_i(\val_i)^+ \indics]}{\E[\val_i \indics]}.
\end{equation}
Combining~\eqref{eq:hyper-reg.1} with~\eqref{eq:hyper-reg.2}
and recalling that the right side of~\eqref{eq:hyper-reg.1}
is at least $\frac{1}{c}$, we have derived:
\begin{align}
\E[\vv_i(\val_i)^+ \indics] & \geq \tfrac{1}{c} \E[\val_i \indics].
\label{eq:hyper-reg.3}
\end{align}
Summing over all $i \in S$ and using the notations
$\vv^+(S) = \sum_{i \in S} \vv_i(\val_i)^+, \,
\val(S) = \sum_{i \in S} \val_i$, we obtain
\begin{equation} \label{eq:hyper-reg.4}
\E[\vv^+(S) \indics] \geq \tfrac{1}{c} \E[\val(S) \indics].
\end{equation}
Finally, summing over all feasible sets $S \in \feas$,
we find that
\begin{equation} \label{eq:hyper-reg.5}
\E[\vv^+(\opt)] \geq \tfrac{1}{c} \E[\val(\opt)].
\end{equation}
The right side is the expected welfare of the efficient
allocation. The left side is the expected revenue of the
mechanism that selects the efficient allocation and
then removes agents whose virtual value is negative,
i.e.\ the VCG-L mechanism that was defined at the end
of Section~\ref{sec:prelims}.
 The expected revenue of
the optimal mechanism is at least as great as that
of VCG-L, so our theorem is proved.
\end{proof}

Based on the proof, we immediately have the following corollary:
\begin{corollary}
If $\feas$ is downward-closed, and all agents
have $c$-bounded hyper-regular distributions,
then the revenue of the VCG-L Mechanism is at
least $1/c$ times the expected welfare
of the efficient allocation.
\end{corollary}

\OMIT{
To state our generalization, we require the following
definition.
\begin{definition}  \label{def:comon}
If $Y,Z$ are two random variables on a probability
space $\Omega$, we say that $Z$ is \emph{stochastically
comonotonic} with $Y$ if the function
$z(y) = \E[Z | Y=y]$ is a monotonically non-decreasing
function of $y$.
\end{definition}
The following theorem is essentially equivalent to
Chebyshev's Integral Inequality; we provide a self-contained
proof here.
\begin{lemma} \label{lem:chebyshev}
If $Y,Z$ are random variables on a probability space, and
$Z$ is  stochastically comonotonic with $Y$,
then
\[ \E[YZ] \geq \E[Y] \cdot \E[Z]. \]
\end{lemma}
\begin{proof}
Let $\bar{y} = \E[Y]$, and let $z(y) = \E[Z | Y=y]$
as in Definition~\ref{def:comon}. Recall that
$z(y)$ is non-decreasing; in particular, for all $y$,
\begin{align*}
(y-\bar{y}) \, (z(y) - z(\bar{y})) \geq 0 \\
(y - \bar{y}) \, z(y) \geq (y - \bar{y}) \, z(\bar{y}).
\end{align*}
Now we have:
\begin{align*}
\E[YZ] - \E[Y] \cdot \E[Z] =
\E[(Y-\bar{y}) \, Z] & =
\E[ \E[ (Y-\bar{y}) \, Z \given Y ] ] \\
& = \E[ (Y-\bar{y}) \E[Z \given Y] ] \\
& = \E[ (Y-\bar{y}) \, z(Y) ] \\
&\geq \E[ (Y-\bar{y}) \, z(\bar{y}) ]
= \E[Y-\bar{y}] \, z(\bar{y}) = 0,
\end{align*}
which proves the lemma.
\end{proof}
}

\subsection{Ratio for non-hyper-regular distributions}
We use an example to show that
even in the setting of a single item auction with $n$ i.i.d.\ bidders,
the revenue-to-welfare ratio for $c$-bounded regular distributions
that are not hyper-regular
may tend to zero as $n$ grows to infinity. 
Defining $\delta>0$ by the equation
$\delta \ln^2 \delta = 1$, our distribution has
cumulative distribution function
$$F(x-\delta)=1-\frac{1}{x\ln^2 x}.$$
(Our choice of $\delta$ is to ensure that $F(x)\geq 0$ for all $x \geq 0$.)
A random variable $X$ with this distribution satisfies
$\E[X]=1/(\ln \delta) < 1.5$ while $\optrev(X) > 0.25$, 
so the distribution is $c$-bounded for any $c \geq 6$.

By computing the density
$$f(x-\delta)=(1-F(x-\delta))'=\frac{\ln x+2}{x^2 \ln^3 x},$$
we find that
$$
\phi(x-\delta)=x-\delta- \frac{1-F(x-\delta)}{f(x-\delta)}
=x-\delta- \frac{x\ln x}{\ln x+2}=-\delta+\frac{2x}{\ln x+2}.
$$
Note that $\phi$ is an increasing function; the distribution
is regular.
Let $$\cutoff = F^{-1} ( \tfrac{n-1}{n} ) \; > \; \frac{n}{\ln^2(n)} - \delta.$$
If $X_1,\ldots,X_n$ are
i.i.d.\ random variables with distribution $F$ and $X^* = 
\max \{X_1,\ldots,X_n\}$ then the event $X^* < \cutoff$ has 
probability $(\frac{n-1}{n})^{n} < 1/e$. By Myerson's Lemma,
the revenue of the optimal mechanism equals $\E[\phi(X^*)^+]$.
An upper bound on this quantity can be derived as follows.
\begin{align*}
\E[\phi(X^*)^+] &= 
\E[\phi(X^*)^+ \, \indith] \;\; +  \;\;
\E[\phi(X^*)^+ \, \indithc] \\
&\leq
\E[\phi(X^*)^+ \, \indith] \;\; + \;\; \phi(\cutoff)^+ \Pr(\cutoff > X^*) \\
&=
\E [\phi(X^*)^+ \, \indith] \;\; +  \;\; \tfrac{1}{e-1} \, \phi(\cutoff)^+ \Pr(\cutoff \leq X^*) \\
&=
\E \left[ \left(\phi(X^*)^+ + \tfrac{1}{e-1} \, \phi(\cutoff)^+ \right) \, \indith
\right] \\ 
&\leq
\frac{e}{e-1} \, \E[\phi(X^*)^+ \, \indith] \\
& \leq
\frac{e}{e-1} \, \E \left[ \frac{2 X^*}{\ln(X^* + \delta) + 2} \, \indith \right] 
\qquad \qquad \mbox{\em (since $\phi(x) > \frac{2x}{\ln(x+\delta)+2}$ for all $x$)}\\
& \leq
\frac{e}{e-1} \cdot \frac{2}{\ln(\cutoff + \delta) + 2} \, \E[X^* \, \indith] \\
& \leq
\frac{e}{e-1} \cdot \frac{2}{\ln(n/\ln^2(n)) + 2} \, \E[X^*]
\end{align*}
The revenue-to-welfare ratio $\E[\phi(X^*)^+] / \E[X^*]$
therefore converges to zero as $n \to \infty$.

\section{Acknowledgement}
We wish to express our gratitude to Anna Karlin for  many
influential conversations about this work.

\appendix

\section{Appendix: Proof of Proposition~\ref{prop:lob.2}}
\label{app:prob}

We begin this appendix with a restatement and 
proof of Lemma~\ref{lem:mdm}.
\begin{lemma} 
For any random variable $X$ and any
constant $a$,
\begin{equation} \label{eq:mdm-shift}
\mdm(X-a) = \mdm(a-X) = \mdm(X) 
\quad \mbox{and} \quad 
\md(X-a) = \md(a-X) = \md(X).
\end{equation}
Furthermore,
\begin{equation} \label{eq:mdm}
\mdm(X) \leq \md(X) \leq 2 \mdm(X).
\end{equation}
\end{lemma}
\begin{proof}
The relations~\eqref{eq:mdm-shift} are immediate
from the definitions.
\OMIT{
To prove $\mdm(X) \leq \md(X)$,
we consider the function $f(a) = \E|X-a|$. For any
fixed $x$, the derivative of $|x-a|$ with respect to $a$
is $-1$ if $x<a$ and $+1$ if $x > a$. Thus, for any
$a$ such that $\Pr(X=a) = 0$, the derivative of
}
To prove the inequalities in \eqref{eq:mdm}, 
it suffices to consider the case when $m(X) \leq 0 = \E X$,
since the general case can then be derived by setting
$a=\E X$ and applying the relations~\eqref{eq:mdm-shift}.
Let $m=m(X)$.
Under the hypothesis that $m \leq 0 = \E X$, we have
\[
\md(X) = \E |X| = 2 \E[X^+] \leq 2 \E[(X-m)^+] \leq 2 \mdm(X),
\]
which establishes one of the two inequalities in~\eqref{eq:mdm}.
To prove the other one, we use the relation $|x| = x \cdot \sgn(x)$
to obtain
\begin{align*}
\md(X) - \mdm(X) & = \E[X \cdot \sgn(X) - (X - m) \cdot \sgn(X-m)] \\ 
& = \E[X \cdot (\sgn(X) - \sgn(X-m))] \;\; + \;\; m \cdot \E[\sgn(X-m)] \\
& = \E [X \cdot (\sgn(X) - \sgn(X-m))] \;\; \geq \;\; 0,
\end{align*}
where the last inequality follows because $\sgn(X) - \sgn(X-m)$
is non-zero only when $m \leq X \leq 0$, in which case both $X$
and $\sgn(X) - \sgn(X-m)$ are non-positive.
\end{proof}

We now commence the proof of Proposition~\ref{prop:lob.2}.
We first need a definition and some preliminary lemmas.

\begin{definition}
If $\eps$ is a $\{\pm 1\}$-valued random variable
we say that $\eps$ is \emph{medially coupled to $X$}
if $\Pr(\eps=1) = \Pr(\eps=-1) = \frac12$ and
$\eps (X - m(X))$ is always non-negative. Note that
this is equivalent to saying that $\eps = \sgn(X-m(X))$
almost surely, except in case the event $X-m(X)=0$ has
positive probability.
\end{definition}

\begin{lemma}
If $\eps$ is medially coupled to $X$ then
\begin{equation} \label{eq:medcoup}
\E[X \given \eps=1] - \E[X \given \eps=-1] = 2 \mdm(X).
\end{equation}
\end{lemma}

\begin{proof}
To prove~\eqref{eq:medcoup}, it suffices to observe
that
\begin{equation}
\E[X \given \eps=1] - \E[X \given \eps=-1] = 2 \cdot \E[\eps X] = 2 \E |X - m(X)|.
\end{equation}
\end{proof}

The proof of the following lemma uses the
technique introduced by Erd\H{o}s in his
proof of Theorem~\ref{thm:littlewood-offord}.

\begin{lemma} \label{lem:lob.1}
Suppose $X_1,\ldots,X_n$ and $\eps_1,\ldots,\eps_n$ are
two $n$-tuples of random variables such that for all $i$, $\mdm(X_i) \geq 1$
and $\eps_i$
is medially coupled to $X_i$. Suppose the
coupled pairs
$\{(X_i,\eps_i)\}_{i=1}^n$ are mutually independent of
one another.
For every sign vector $\sigma \in \{\pm 1\}^n$, let
$$
E(\sigma) = \E[X_1 + \cdots + X_n \given (\eps_1,\ldots,\eps_n)=\sigma].
$$
If $I$ is any half-open interval of length 2,
the number of sign vectors $\sigma$ such that
$E(\sigma) \in I$ is at most $\binom{n}{\lfloor n/2 \rfloor}$.
\end{lemma}
\begin{proof}
For any $i$ we have
  $$\E[X_i \given (\eps_1,\ldots,\eps_n)=\sigma] =
     \E[X_i \given \eps_i=\sigma_i],$$
since $\eps_j$ is independent of $\eps_i$ for $j \neq i$.
Thus,
  $$E(\sigma) = \sum_{i=1}^n \E[X_i \given \eps_i = \sigma_i].$$
For any two sign vectors $\sigma,\sigma'$, we have
\begin{align*}
E(\sigma)-E(\sigma') &= \sum_{i=1}^n \left(
\E[X_i \given \eps_i = \sigma_i] - \E[X_i \given \eps_i = \sigma'_i]
\right) \\
&= \sum_{i=1}^n (\sigma_i - \sigma'_i) \mdm(X_i),
\end{align*}
where the last line follows from~\eqref{eq:medcoup}.
Thus, if  $\sigma \succ \sigma'$ (meaning, $\sigma_i \geq \sigma'_i$
for all $i$, and the inequality is strict for at least one $i$)
it follows that
$E(\sigma)-E(\sigma') \geq 2$.
The lemma now follows from Sperner's Lemma~\citep{sperner},
which states that any collection of more than
$\binom{n}{\flr{n/2}}$ sign vectors contains a
pair such that $\sigma \succ \sigma'$.
\end{proof}

Recall the statement of Proposition~\ref{prop:lob.2}.
\begin{proposition}
If $X_1,\ldots,X_n$ are independent random variables
and $\mdm(X_i) \geq 1$ for all $i$, then
$\mdm(X_1+\ldots+X_n) \geq \frac{1}{12} \sqrt{n}.$
\end{proposition}
\begin{proof}
Let $\eps_1,\ldots,\eps_n$ be independent random variables
medially coupled to $X_1,\ldots,X_n$.
Let $X=X_1+\ldots+X_n$ and $m = m(X)$. Jensen's convex function inequality
applied to the random variable $|X-m|$ implies
\[
\E |X-m| = \E_{\sigma} \left[
\E ( |X-m| \; \given \; \sigma ) \right] \geq
\E_{\sigma} \left[ |E(\sigma) - m| \right].
\]
We will prove a lower bound on the quantity
appearing on the right-hand side.
Let $k = \lceil \frac{\sqrt{n}}{3} \rceil -1 $.
 and apply
Lemma~\ref{lem:lob.1} to the intervals
$I_j = (m+2j-1,m+2j+1]$ for $-k \leq j \leq k$.
Each of these intervals contains at most $\binom{n}{\flr{n/2}}$
of the numbers $E(\sigma)$. Hence, there are at most
$(2k+1) \binom{n}{\flr{n/2}}$
sign vectors $\sigma$ such that
$m-2k-1 < E(\sigma) \leq m+2k+1$.
The inequality
\[
(2k+1) \binom{n}{\flr{n/2}} =
\left(2 \cel{\tfrac{\sqrt{n}}{3}} - 1\right) \binom{n}{\flr{n/2}} <
\frac34 2^n
\]
can be verified by exhaustive enumeration over small
values of $n$ combined with the asymptotic estimate
$\binom{n}{\flr{n/2}} \sim \sqrt{\frac{2}{\pi n}} \cdot 2^n$
as $n \to \infty$. Therefore, we have
\[
\E_{\sigma} [|E(\sigma)-m|] \geq (2k+1) \Pr_\sigma(|E(\sigma)-m| \geq 2k+1)
\geq \frac14 \left(2 \cel{\tfrac{\sqrt{n}}{3}} - 1 \right)
\geq \frac14 \cel{\tfrac{\sqrt{n}}{3}}
\geq \tfrac{1}{12} \sqrt{n}.
\]
\BKnote{It's a problem that we have the trailing $\frac14$ in this
theorem. Try to get simply a constant.}
\end{proof}
\end{document}